\documentclass{llncs}
\usepackage{amsmath,amssymb}
\usepackage{times}
\usepackage{tikz}
\usepackage{fixltx2e}
\usepackage{longtable}
\usepackage{xspace}
\usepackage{listings}
\usepackage{graphicx}
\usepackage[english]{babel}
\usepackage{amssymb}
\usepackage{amsmath}
\usepackage{graphics}
\usepackage[ansinew]{inputenc}
\usepackage{lscape}
\usepackage{longtable}
\usepackage{url}
\usepackage{float}
\usepackage{color,soul}
\usepackage{tikz}
\usepackage{pgflibraryshapes}
\usepackage{pgflibrarysnakes}
\usepackage{pgffor}
\usepackage[pdftex,bookmarks=true]{hyperref}
\usepackage{cite}
\usepackage{mathtools}

\DeclareMathOperator{\rowsp}{rowsp}

\newcommand{\CP}[1]{}

\newcommand{\cost}[1]{{\textcolor{blue}{}}}

\DeclareMathOperator{\sd}{sd}
\DeclareMathOperator{\affn}{Aff_n}
\DeclareMathOperator{\gln}{GL_n}
\DeclareMathOperator{\reg}{reg}
\DeclareMathOperator{\sat}{sat}
\newcommand{\Mod}[1]{\ \mathrm{mod}\ #1}

\newtheorem{Definition}{Definition}
\newtheorem{Notation}{Notation}

\newtheorem{Theorem}{Theorem}
\newtheorem{Proposition}[Theorem]{Proposition}
\newtheorem{Lemma}[Theorem]{Lemma}

\renewenvironment{proof}{\textsc{Proof:}} {$\Box$\\}

\newcommand{\NN}{\ensuremath{\mathbb{N}}\xspace}

\newcommand{\FF}{\ensuremath{\mathbb{F}}\xspace}

\newcommand{\groebner}{Gr\"obner\ }

\title{Stronger bounds on the cost of computing \groebner bases for HFE systems}

\author{Elisa Gorla\inst{1} \and Daniela Mueller\inst{2}\thanks{Research supported by a Postgraduate Government of Ireland Scholarship from the Irish Research Council.} \and Christophe Petit\inst{3}}

\institute{Institut de Math\'ematiques, Universit\'e de Neuch\^atel, Switzerland
\and School of Mathematics and Statistics, University College Dublin, Ireland 
\and School of Computer Science, University of Birmingham, United Kingdom}

\begin{document}
\pagestyle{plain}

\maketitle

\begin{abstract}
We give upper bounds for the solving degree and the last fall degree of the polynomial system associated to the HFE (Hidden Field Equations) cryptosystem. Our bounds improve the known bounds for this type of systems.
We also present new results on the connection between the solving degree and the last fall degree and prove that, in some cases, the solving degree is independent of coordinate changes.
\end{abstract}

\section{Introduction}

Multivariate cryptography is one of a handful of proposals of post-quantum cryptosystems, i.e. cryptosystems that would remain secure even in the presence of a quantum computer.
In multivariate cryptography, the hard computational problem that one has to solve in order to retrieve the original message from the ciphertext is solving a system of multivariate polynomial equations over a finite field.

%

\groebner bases are a widely used tool for solving systems of polynomial equations. In positive characteristic, along with SAT solvers, they are essentially the only tool of general applicability at our disposal, since we have no numerical methods available. Bounds on the complexity of computing a (degree reverse lexicographic) \groebner basis of a system of polynomial equations associated to a given cryptosystem provide bounds on the complexity of recovering the secret message from the ciphertext, hence bounds on the security of the cryptosystems. Such bounds are of fundamental importance, as they give us an indication on how to choose the parameters of the cryptosystem in order to achieve the desired level of security.

Beyond Buchberger's Algorithm, a family of algorithms based on linear algebra are available for computing \groebner bases. They are based on an idea of Lazard~\cite{Lazard} and many variations of such algorithms are currently used. This family of algorithms includes $F_4/F_5$~\cite{F4,F5} and XL/Mutant XL~\cite{CKPS00}, whose complexity is measured by the \emph{solving degree}. For other variations of these algorithms, the complexity is measured by the last fall degree, first introduced in ~\cite{HKYCrypto}.

In this paper, we study the systems associated to the cryptosystem HFE in its basic version, as it was proposed by Patarin in~\cite{P}. It has been experimentally observed that \groebner basis algorithms perform much better on HFE systems than on generic systems~\cite{FJ}. More precisely, the solving degree appeared to be much smaller for HFE systems than for generic systems with the same number of variables and the same degrees.

In this paper, we provide upper bounds on the solving degree and on the last fall degree of such systems, which improve on the known ones. Our bound on the solving degree relies on results from~\cite{CG}, which connect the solving degree of a system of inhomogeneous equations to the Castelnuovo-Mumford regularity of the ideal generated by the homogenization of the equations. In particular, our bound on the solving degree relies on~\cite[Theorem~3.14, Theorem~3.23, and Corollary~3.25]{CG} and on a new estimate of the degree of the equations obtained by fake Weil descent, which we prove in Lemma~6 and Remark~7. Differently from other bounds on the solving degree that appeared previously in the literature, the bounds that we obtain are rigorously proved, meaning that our approach does not rely on any unproved assumptions or heuristics. 

Our bound on the last fall degree is inspired by the bound from~\cite{HKYY}.
Our main result is Theorem~\ref{main}, where we prove that the last fall degree of the Weil descent system of a set of polynomials over $\mathbb{F}_{q^n}$ is bounded by $(q-1)\lceil\log_q(d)+1\rceil+1$ for a certain $d$. When applied to HFE systems, this improves the bound given in~\cite{HKYY} by approximately a factor of 2. The improvement in the bound is a consequence of Lemma~8, where we obtain a bound which is tighter than in previous work. Our tighter bound allows us to produce a more precise estimate of the degrees in which certain polynomials are computed by a linear-algebra based \groebner basis algorithm, allowing us to get a tighter grip on the last fall degree of the system.

\subsection{Organisation of the paper}

The paper is organised as follows. In Section~\ref{sec:pre}, we recall the definitions of the solving degree and last fall degree and introduce the notation. In Section~\ref{sec:RelateDegrees}, we relate the solving degree and the last fall degree. In Section~\ref{sec:coord}, we show that the solving degree is invariant under coordinate changes, for systems which have a single solution of multiplicity one over the algebraic closure. In Section~\ref{sec:HFE}, we give a provable bound for the solving degree of HFE systems. Finally,  in Section~\ref{sec:betterbound} we prove our main theorem, which gives a better bound on the last fall degree for Weil descent systems, and in particular a better bound for HFE sytems. 


\section{Preliminaries and notation}\label{sec:pre}

Let $k$ be a field, let $R=k[X]$ and $S=k[X_0,\ldots,X_{n-1}]$ be polynomial rings with coefficients in $k$. 
Let $R_{\leq d}$ and $S_{\leq d}$ be the $k$-linear spaces of polynomials of degree $\leq d$ in $R$ and $S$, respectively. For an ideal $I\subseteq S$, denote by $I_{\leq d}$ the vector space $I\cap S_{\leq d}$. For polynomials $f,g\in R$, write $f\Mod g$ for the remainder of the division of $f$ by $g$. 
Let $\mathcal{E}$ be a finite subset of $S$, and let $I$ be the ideal generated by $\mathcal{E}$. Let $\deg(\mathcal{E})=\max\{\deg(f) : f\in \mathcal{E}\}$.

\subsection{Last fall degree}

We will closely follow the notations in \cite{HKYY} and briefly recall some of its definitions, 
which we will need later. The reader is refered to \cite{HKYY} for more details.

\begin{Definition}
For $i\in\mathbb{Z}_{\geq0}$, we let $V_{\mathcal{E},i}$ be the smallest $k$-vector space such that
\begin{itemize}
 \item $\mathcal{E}\cap S_{\leq i}=\{f\in\mathcal{E} : \deg(f)\leq i\}\subseteq V_{\mathcal{E},i}$;
 \item if $g\in V_{\mathcal{E},i}$ and if $h\in S$ with $\deg(gh)\leq i$, then $gh\in 
 V_{\mathcal{E},i}$.
\end{itemize}
If $\mathcal{E}$ is fixed, we write $V_i$ instead of $V_{\mathcal{E},i}$. 
Let $V_{\mathcal{E},\infty}=\cup_{i\geq 0} V_{\mathcal{E},i}$.
\end{Definition}

Intuitively, we construct $V_i$ from $\mathcal{E}$ by doing ideal operations which only involve polynomials of degree at most $i$. 

\begin{remark}
It is easy to show that $V_{\mathcal{E},\infty}=I$.
\end{remark}

\begin{Notation}
For $g,h\in S$ and $i\in\mathbb{Z}_{\geq0}$, we write $g\equiv_{\mathcal{E},i}h$ if $g-h\in V_{\mathcal{E},i}$. If $\mathcal{E}$ is fixed, we write $g\equiv_i h$.
\end{Notation}

\begin{Definition}\label{LastFallDegree}\cite{HKYCrypto}
Let $\mathcal{E}$ be a finite subset of $S$ and let $I=(\mathcal{E})$ be the ideal generated by $\mathcal{E}$. The minimal $d\in\mathbb{Z}_{\geq0}\cup\{\infty\}$ such that for all $f\in I$ we have $f\in V_{\max\{d,\deg(f)\}}$ is called the last fall degree of $\mathcal{E}$ and is denoted by $d_\mathcal{E}$.
\end{Definition}

Notice that neither the last fall degree nor the vector spaces $V_{\mathcal{E},d}$ depend on the choice of a term order.

\medskip

We now briefly recall how the last fall degree is related to the complexity of solving a system $\mathcal{E}\subset S$. Suppose that the ideal generated by $\mathcal{E}$ is radical and let $e$ be the number of solutions of $\mathcal{E}$ over $\overline{k}$. Suppose also that one can factor a polynomial of degree $t$ in a number of field operations that is polynomial in $g(t)$, for some given function $g$ of $t$. In~\cite[Proposition~2.3 and Proposition~2.8]{HKYY} the authors outline an algorithm that computes the solutions of $\mathcal{E}$ in a number of field operations which is polynomial in $(m+d)^d$, $g(d)$, and the size of $\mathcal{E}$, where $d=\max\{d_\mathcal{E},e\}$.

\subsection{Solving degree}\label{solvdeg}

For the ease of the reader, we describe the computations carried on by a linear algebra-based \groebner basis algorithm as in \cite{Lazard}.
Fix a term order $\sigma$ and a degree $d\geq 1$. Let $[d]$ denote the set $\{0,\ldots,d\}$ and let 
$$M_d=\{a\in [d]^n\mid a_0+\ldots+a_{n-1}\leq d\}.$$ 
The elements of $M_d$ correspond to the monomials in $S_{\leq d}$ via 
$$a=(a_0,\ldots,a_{n-1})\longleftrightarrow X_0^{a_0}\cdots X_{n-1}^{a_{n-1}}.$$

Build a matrix $M$ whose columns are indexed by the elements of $M_d$ in decreasing order from left to right with respect to $\sigma$. The rows correspond to polynomials of the form $uf$ where $u\in S$ is a monomial, $f\in\mathcal{E}$, and $\deg(uf)\leq d$. Notice that this includes the possibility that $u=1$. In order to associate a row $r(g)$ to a polynomial $g$, write $$g=\sum_{a=(a_0,\ldots,a_{n-1})\in M_d} \alpha_a X_0^{a_0}\cdots X_{n-1}^{a_{n-1}},$$ then $r(g)_a=\alpha_a$. 

Perform Gaussian elimination on $M$ to obtain a matrix in reduced row echelon form. 
Any row $r=(r_{a}\mid a\in M_d)$ in the reduced row echelon form of $M$ corresponds to a polynomial 
$$f_r=\sum_{a\in M_d}r_a X_0^{a_0}\cdots X_{n-1}^{a_{n-1}}.$$ If $\deg(f_r)<d$, we add new rows to $M$ corresponding to polynomials of the form $uf_r$, where $u$ is a monomial, $\deg(uf_r)\leq d$ and $uf_r\not\in\rowsp(M)$. Here $\rowsp(M)$ denotes the rowspace of $M$. Repeat the computation of the reduced row echelon form and the operation of adding new rows, until there are no new rows added. 

\begin{Notation}\label{wd}
For any $d\geq 1$ let $W_d$ be the vector space generated by the rows of $M$ after running the algorithm that we just described. 
\end{Notation}

It is clear that $W_d\subseteq I_{\leq d}$. For a given $d$, one may have $W_d\neq I_{\leq d}$. However, it is well-known that $W_d=I_{\leq d}$ for $d\gg 0$. In particular, $W_d$ contains a \groebner basis of $I$ with respect to $\sigma$ for $d\gg 0$.

\begin{Definition}\label{sd}
The {\bf solving degree} of $\mathcal{E}$ with respect to a term order $\sigma$ is the least $d$ such that $W_d$ contains a \groebner basis of $I$ with respect to $\sigma$. We denote it by $\sd_{\sigma}(\mathcal{E})$.
\end{Definition}

\begin{remark}
Notice that the elements of a reduced \groebner basis of $I$ appear as rows of the matrix obtained from $M$ by running the algorithm described above. In fact, since the matrix is in reduced row echelon form, then no cancellation is possible among the leading terms of different rows. In particular, the leading terms of the elements of $W_d$ are exactly the leading terms of the rows of the matrix. Therefore, if $W_d$ contains a \groebner basis of $I$ with respect to $\sigma$, then there is a set of rows of the matrix which forms a \groebner basis of $I$. Take a minimal set of such rows. Since the matrix is in reduced row echelon form, they form a reduced \groebner basis of $I$.
\end{remark}

\subsection{Weil descent}

Again, we closely follow the setup and notations from \cite{HKYY}. Let $q$ be a prime power and $n\in\mathbb{Z}_{\geq1}$. Let $k$ be a finite field of cardinality $q^n$ and let $k'$ be its subfield of cardinality $q$. Let $R=k[X]$ be a polynomial ring and let $\mathcal{F}$ be a finite subset of $R$. Let $\alpha_0,\ldots,\alpha_{n-1}$ be a basis of $k/k'$. 

\begin{Definition}
Write $X=\sum^{n-1}_{j=0}\alpha_jX_j$. For $f\in\mathcal{F}$ we define $[f]_j\in k'[X_0,\ldots,X_{n-1}]$, $j=0,\ldots,n-1$, by 
$$f\left(\sum^{n-1}_{j=0}\alpha_jX_j\right)\equiv\sum^{n-1}_{j=0}[f]_j\alpha_j\Mod (X^q_j-X_j\mid j=0,\ldots,n-1)$$ 
with $[f]_j$ of minimal degree, i.e. $\deg_{X_i}([f]_j)\leq q-1$ for all $i$. 
Let $$\mathcal{F}'=\{[f]_j : f\in\mathcal{F}, j=0,\ldots,n-1\}$$ be the {\bf Weil descent system} of $\mathcal{F}$ with respect to the basis $\{\alpha_j\}$, and let $$\mathcal{F}'_f=\mathcal{F}'\cup\{X^q_j-X_j : j=0,\ldots,n-1\}.$$ 
Let $$\mathcal{F}_f=\mathcal{F}\cup\{X^{q^n}-X\}.$$
\end{Definition}

\begin{remark}
We have followed here the notation of \cite{HKYY}. The subscript $f$ refers to the \textit{field} equations added.
\end{remark}

Now we recall the fake Weil descent defined in \cite{HKYCrypto} and \cite{HKYY}. Unlike the Weil descent system, this system is defined over the larger field $k$. 

\begin{Definition}
Let $X^{e'}=X^e\Mod X^{q^n}-X$ and write $e'=\sum_{j=0}^{n-1}e'_jq^j$ in base $q$ expansion with $e'_j\in\{0,1,\ldots,q-1\}$. We let $$\overline{X^e}=X_0^{e'_0}\ldots X_{n-1}^{e'_{n-1}}\in S=k[X_0,\ldots,X_{n-1}].$$ 
This definition can be extended $k$-linearly to all polynomials in $R$ and gives a map from $R$ to $S$. For any $f\in R$, we denote by $\bar{f}\in S$ the image of $f$ via this map. 
Let $$\overline{\mathcal{F}}=\{\overline{f} : f\in\mathcal{F}\}$$ be the {\bf fake Weil descent system} of $\mathcal{F}$. Let $$\overline{\mathcal{F}}_f=\overline{\mathcal{F}}\cup\{X^q_0-X_1,\ldots,X^q_{n-2}-X_{n-1},X^q_{n-1}-X_0\}.$$
\end{Definition}

\begin{remark}\label{degrees}
One has $\deg(\overline{f})=\max\{\deg([f]_j) : j=0,\ldots,n-1\}$.
\end{remark}

\begin{remark}\label{RelateWeil}
\cite[Proposition~4.1]{HKYY} relates the last fall degree of the two types of Weil descents by showing that
$$\max\{d_{\mathcal{F}'_f},q,\deg(\mathcal{F}')\}\leq\max\{d_{\overline{\mathcal{F}}_f},q,\deg(\mathcal{F}')\}.$$
We are interested in finding an upper bound for the last fall degree and therefore mostly work with the system $\overline{\mathcal{F}}_f$.
\end{remark}


\section{Relating the solving degree and last fall degree}\label{sec:RelateDegrees}

In this section, we clarify the relationship between solving degree and last fall degree, for degree-compatible term orders.

\begin{Theorem}\label{mm}
Let $\mathcal{E}\subset S$ be a finite set of polynomials and let $\sigma$ be a degree-compatible term order. Let $W_d$ be the vector space constructed as in Notation~\ref{wd}. Then $V_{\mathcal{E},d}=W_d$ for all $d\geq0$. Moreover $$\sd_{\sigma}(\mathcal{E})\geq d_{\mathcal{E}}.$$
\end{Theorem}

\begin{proof}
Since $W_d\supseteq \mathcal{E}\cap S_{\leq d}$ and the operations performed by the algorithm described in Section~\ref{solvdeg} only involve polynomials of degree at most $d$, 
by definition $W_d\subseteq V_{\mathcal{E},d}$. 
We now prove the reverse inclusion. By Definition~\ref{LastFallDegree} and since 
$W_d\supseteq \mathcal{E}\cap S_{\leq d}$, it suffices to show that if $g\in W_d$ and $h\in S$ with 
$\deg(gh)\leq d$, then $gh\in W_d$. We may assume without loss of generality that $h$ is a monomial.
The rows of the matrix in reduced row echelon form, say $M$, produced by the algorithm described in Section~\ref{solvdeg} are a basis of $W_d$ by definition. Since $\sigma$ is degree compatible and no cancellation among leading terms is possible, then the rows corresponding to polynomials of degree smaller than $d$ are a basis of $W_d\cap S_{<d}$. Let $g\in W_d\cap S_{<d}$, let $h\in S_{\leq d-\deg(g)}$ be a monomial. Then $g$ is a linear combination of some rows of $M$, say $r_1,\ldots,r_\ell$. Since the algorithm terminated, then $hr_i\in W_d$ for each $i=1,\ldots,\ell$, 
hence $hg\in W_d$.

In order to show that $\sd_{\sigma}(\mathcal{E})\geq d_{\mathcal{E}}$, it suffices to show that for all $f\in(\mathcal{E})$ one has $f\in V_{\mathcal{E},\max\{\sd_{\sigma}(\mathcal{E}),\deg(f)\}}$. Let $g_1\ldots,g_\ell$ be a \groebner basis of $(\mathcal{E})$ with respect to $\sigma$. Then $$f=\sum_{i=0}^\ell h_ig_i,$$ with $\deg(h_i)+\deg(g_i)\leq \deg(f)$ for all $i$. By definition of solving degree and since $V_{\mathcal{E},d}=W_d$ for all $d$, then $g_i\in V_{\mathcal{E},\sd_{\sigma}(\mathcal{E})}$ for all $i$. Therefore, $f\in V_{\mathcal{E},\max\{\sd_{\sigma}(\mathcal{E}),\deg(f)\}}$.
\end{proof}

Notice that it is possible that $\sd_{\sigma}(\mathcal{E})>d_{\mathcal{E}}$.

\begin{example}
Let $\mathcal{E}=\{g\}$ consist of a single polynomial. Then $\sd_{\sigma}(\mathcal{E})=\deg(g)$ for any term order $\sigma$, but $d_{\mathcal{E}}=0$. In fact, for any $f\in S$ one has $fg\in V_{\max\{0,\deg(fg)\}}=V_{\deg(fg)}$.
\end{example}

Moreover, the conclusion of Theorem~\ref{mm} is false in general for term orders which are not degree-compatible.

\begin{example}
Let $d\geq 2$ be an integer. Let $\mathcal{E}=\{X_0-X_0X_2^{d-1},X_1-X_2^d\}$ and let $\sigma$ be the lexicographic order on $k[X_0,X_1,X_2]$ with $X_0>X_1>X_2$. The elements of $\mathcal{E}$ are a \groebner basis of the ideal that they generate, since their leading terms are coprime. Therefore, $\sd_{\sigma}(\mathcal{E})=d$. Let $$f=X_0X_2-X_0X_1=X_2(X_0-X_0X_2^{d-1})-X_0(X_1-X_2^d)\in V_{d+1}.$$ Since $f\not\in V_d=\langle X_0-X_0X_2^{d-1},X_1-X_2^d \rangle$, then $$d_{\mathcal{E}}\geq d+1>\sd_{\sigma}(\mathcal{E}).$$
\end{example}

\section{Solving degree and coordinate changes}\label{sec:coord}

One difficulty in estimating the solving degree comes from the fact that the degrees of the elements of a reduced \groebner basis of $I$ may vary with the term order. While many results from commutative algebra allow us to provide estimates for the solving degree of polynomial systems in generic coordinates (see e.g.~\cite[Section~3.3]{CG}), results that hold for a given (non generic) system of coordinates are often much harder to prove and require ad-hoc arguments. With this in mind, in this section we find a sufficient condition for the solving degree to be independent of coordinate changes. It turns out that it suffices to assume that the system has a unique solution (of multiplicity one) over the algebraic closure.

Let $S=k[X_0,...,X_{n-1}]$, let $K\supseteq k$ be a field extension and let $\varphi\in\affn(K)$ be a change of coordinates over $K$, where $\affn(K)=K^n\rtimes\gln(K)$ is the affine group of degree $n$ over the field $K$. $\affn(K)$ is isomorphic to the group of maps $\{x\rightarrow Ax+b\}$ where $A$ is an invertible $n\times n$ matrix (over $K$) and $b$ is an element of $K^n$. The goal of this section is showing that the solving degree with respect to a degree-compatible term order does not depend on the system of coordinates, for systems $\mathcal{E}$ which have a simple zero.

\begin{Definition}
We say that $\mathcal{E}\subset S$ has {\bf a single solution of multiplicity one} over the algebraic closure, or {\bf a simple zero}, if $I=(\mathcal{E})$ is radical and $\mathcal{E}$ has exactly one solution over the algebraic closure $\overline{k}$ of $k$.
\end{Definition}

We concentrate on $I$ not homogeneous, since all the results that we prove are trivial in the homogeneous case. Then $\mathcal{E}$ has a simple zero if and only if the zero locus of $I$ over $\overline{k}$ is a point $(a_0,\ldots,a_{n-1})\in k^n$ and $I=(X_0-a_0,\ldots,X_{n-1}-a_{n-1})$. Equivalently, $\mathcal{E}$ has a simple zero if and only if $I$ contains $n$ linearly independent linear forms. Moreover, $I$ contains $n$ linearly independent linear forms if and only if $I^h=(f^h\mid f\in I)$ is generated by linear forms. 

\begin{Theorem}\label{deg1}
Assume that $\mathcal{E}$ has a simple zero, say $(a_0,\ldots,a_{n-1})\in k^n$, and that $V_{\mathcal{E},d}$ contains $n$ linearly independent linear forms. Then they can be obtained by computing the reduced row echelon form of the Macaulay matrix $M$ of $\mathcal{E}$ in degree $d$ with respect to any degree-compatible term order. 

In particular, $X_0-a_0,\ldots,X_{n-1}-a_{n-1}$ is the reduced \groebner basis of $I$ with respect to any term order and it may be obtained by running the algorithm described in Section~\ref{solvdeg} in degree $d$ with respect to any degree-compatible term order.
\end{Theorem}

\begin{proof}
By assumption $X_0-a_0,\ldots,X_{n-1}-a_{n-1}\in V_{\mathcal{E},d}$. Assume now that we have brought the Macaulay matrix  in degree $d$ in reduced row echelon form. By Theorem~\ref{mm}, $X_i-a_i$ is a linear combination of the rows of the matrix for all $i$. Such a combination cannot involve rows whose leading term is strictly larger than $X_i$. In fact, no cancellation is possible among leading terms of the rows, since the matrix is in reduced row echelon form. In addition, there must be a row that has leading term $X_i$ and this holds for all $i$. Since the term order is degree compatible, each of these rows corresponds to a linear form. Therefore, the last $n$ nonzero rows of the matrix in reduced row echelon form are the polynomials $X_0-a_0,\ldots,X_{n-1}-a_{n-1}$.
\end{proof}

\begin{Lemma}[\cite{HKYY}, Proposition~2.3.iv]\label{gF}
Let $K\supseteq k$ be a field extension and let $\varphi\in\affn(K)$. Then 
$$\varphi(V_{\mathcal{E},d})=V_{\varphi(\mathcal{E}),d}$$ for all $d\in\NN$.
\end{Lemma}

\begin{Theorem}\label{coordchange}
Let $K\supseteq k$ be a field extension. Assume that $\mathcal{E}\subset k[X_0,\ldots,X_{n-1}]$ has a simple zero and let $\sigma$ be a degree-compatible term order. 
Then $$\sd_{\sigma}(\mathcal{E})=\min\{e\in\NN\mid \dim(V_{\mathcal{E},e}\cap S_{\leq 1})\geq n\}.$$ 
In addition, the solving degree of $\mathcal{E}$ does not depend on the choice of a coordinate change defined over $K$, nor on the choice of a degree-compatible term order.
\end{Theorem}

\begin{proof}
Let $\varphi\in\affn(K)$. Since $\mathcal{E}$ has a simple zero over $\overline{k}$, then $I=(X_0-a_0,\ldots,X_{n-1}-a_{n-1})$ for some $a_0,\ldots,a_{n-1}\in k$. Let $\varepsilon(\mathcal{E})$ be the smallest integer $e$ such that $V_{\mathcal{E},e}$ contains $n$ linearly independent linear forms. 
Hence $\varepsilon(\mathcal{E})=\sd_{\sigma}(\mathcal{E})$, by definition of solving degree, Theorem~\ref{mm}, and Theorem~\ref{deg1}. 
By Lemma~\ref{gF}, $\varepsilon(\mathcal{E})=\varepsilon(\varphi(\mathcal{E}))$. Combining all equalities one gets 
$$\sd_{\sigma}(\mathcal{E})=\varepsilon(\mathcal{E})=\varepsilon(\varphi(\mathcal{E}))=\sd_{\tau}(\varphi(\mathcal{E}))$$
for any $\sigma,\tau$ degree-compatible term orders.
\end{proof}

\begin{remark}
Notice that, in general, $\tilde{I}=(f^h\mid f\in\mathcal{E})$ is not generated by linear forms, hence it does not have a simple zero, even under the assumption that $I=(f\mid f\in\mathcal{E})$ does. In particular, the result of Theorem~\ref{coordchange} by itself is not sufficient to conclude that the assumptions of~\cite[Theorem~3.23]{CG} are satisfied and hence conclude that $\sd(I)\leq\reg(\tilde{I})$. Nevertheless, in Section~\ref{sec:HFE} we bypass this problem and prove directly that the systems that interest us have $\tilde{I}$ in generic coordinates, which allows us to apply~\cite[Theorem~3.23]{CG} to bound their solving degree.
\end{remark}

Notice that, when $\mathcal{E}$ does not have a simple zero, the solving degree may depend on the choice of a system of coordinates.

\begin{example}
Let $\mathcal{E}=\{X_0^2, X_1^2\}\subseteq S=\FF_3[X_0,X_1]$ and let $\sigma$ be any term order with $X_0>X_1$. Let $\varphi(X_0)=X_0$ and $\varphi(X_1)=X_0+X_1$. Then 
$$\varphi(\mathcal{E})=\{X_0^2,X_0^2-X_0X_1+X_1^2\}$$ and $$X_1^3=-X_0\cdot X_0^2+(X_0+X_1)(X_0^2-X_0X_1+X_1^2)\in (\varphi(\mathcal{E})).$$
It is easy to check that the reduced \groebner basis of $(\varphi(\mathcal{E}))$ with respect to $\sigma$ is 
$$\left\{X_0^2,X_0X_1-X_1^2,X_1^3\right\},$$ therefore $$\sd_{\sigma}(\mathcal{E})=2<3=\sd_{\sigma}(\varphi(\mathcal{E})).$$
\end{example}

\section{A simple bound for the solving degree of HFE}\label{sec:HFE}

In this section we provide a simple bound for the solving degree of systems of the form $\overline{\mathcal{F}}_f$. These bounds apply in particular to the fake Weil descent system of the basic version of HFE, as proposed in~\cite{P}. Experimental evidence that the solving degree of the Weil descent system of HFE is smaller than that of generic systems was obtained in~\cite{FJ}. A provable bound for the solving degree can be obtained using the techniques from~\cite{CG}. The next theorem is inspired by the proof of~\cite[Theorem~3.26]{CG}.

\begin{Theorem}
Let $k$ be a finite field of cardinality $q^n$ and let $\mathcal{F}=\{f\}\subset R$, $d=\deg(f)$. 
Let $DRL$ denote the Degree Reverse Lexicographic term order. 
Then $$\sd_{DRL}(\overline{\mathcal{F}}_f)\leq\deg(\overline{f})+(q-1)n\leq(q-1)(\lfloor\log_q(d)+1\rfloor+n).$$
In particular, if
$$f=\sum_{i,j}\beta_{i,j}X^{q^{\theta_{ij}}+q^{\varphi_{ij}}}+\sum_\ell \alpha_\ell X^{q^{\zeta_\ell}}+\mu\in k[X],$$ 
then $\deg(\overline{f})\leq 2$, hence
$$\sd_{DRL}(\overline{\mathcal{F}}_f)\leq (q-1)n+2.$$
\end{Theorem}

\begin{proof}
Let $\overline{\mathcal{F}}_f^h=\{f^h\mid f\in\overline{\mathcal{F}}_f\}$ be the system obtained from $\overline{\mathcal{F}}_f$ by homogenizing each equation with respect to $X_n$, where $X_n$ is a new variable. Let $$J=(f^h\mid f\in\overline{\mathcal{F}}_f)\subset S[X_n].$$ We claim that $J$ is in generic coordinates. According to~\cite[Theorem~2.4 and Definition~1.5]{BS}, in our situation $J$ is in generic coordinates if and only if $X_n$ is not a zero divisor on $S[X_n]/J^{\sat}$, where $J^{\sat}$ is the saturation of $J$ with respect to the irrelevant maximal ideal of $S[X_n]$. Substituting $X_n=0$ in $\mathcal{E}^h$ one obtains the equations $X_0=\ldots=X_{n-1}=0$. Therefore the projective zero locus of $J$ does not contain any point with $X_n=0$. This means that $X_n\nmid 0$ modulo $J^{\sat}$, hence proving that $J$ is in generic coordinates.

Denote by $\reg(J)$ the Castelnuovo-Mumford regularity of $J$ (see~\cite[Definition~3.17]{CG} for a definition of Castelnuovo-Mumford regularity). Since $J$ is in generic coordinates, then 
\begin{equation}\label{sd_hfe}
\sd_{DRL}(\overline{\mathcal{F}}_f)\leq\reg(J)\leq (q-1)(\lfloor\log_q(d)+1\rfloor+n),
\end{equation} 
where the first inequality follows from~\cite[Theorem~3.23]{CG}. 
By Remark~\ref{degov} we have $\deg(\overline{f})\leq (q-1)\lfloor\log_q(d)+1\rfloor$.
The second inequality in (\ref{sd_hfe}) now follows from the fact that $\overline{\mathcal{F}}_f$ consists of one equation of degree smaller than or equal to $(q-1)\lfloor\log_q(d)+1\rfloor$ and $n$ equations of degree $q$.
\end{proof}

\section{An improved bound on the last fall degree}\label{sec:betterbound}

In this section we study the last fall degree $d_{\mathcal{F}'_f}$ of the system $\mathcal{F}'_f$. Let $k$ be a finite field of cardinality $q^n$, $R=k[X]$, and $S=k[X_0,...,X_{n-1}]$. 
In~\cite[Theorem 4.5]{HKYY} it is shown that 
\begin{equation}\label{HKYYbound}
d_{\mathcal{F}'_f}\leq \max\left\{\lfloor2(q-1)(\log_q(\deg(\mathcal{F})+1)+1)\rfloor,q\right\}.\end{equation}
As proved in~\cite[Proposition~4.1]{HKYY}, it suffices to bound the last fall degree of the system
$$\overline{\mathcal{F}}_f=\{\overline{f} : f\in\mathcal{F}\}\cup
\{X^q_0-X_1,\ldots,X^q_{n-2}-X_{n-1},X^q_{n-1}-X_0\}.$$ 
In this section, we improve the bound on the last fall degree of $\overline{\mathcal{F}}_f$ proven in~\cite{HKYY} by approximately a factor two. This results in a bound that improves (\ref{HKYYbound}) by the same factor, and ultimately leads to Theorem~\ref{main}.

\begin{Notation}
Throughout the section, we write $\equiv_i$ in place of $\equiv_{\overline{\mathcal{F}}_f,i}$.
\end{Notation}

\begin{Definition}
For $e\in\mathbb{Z}_{\geq0}$, write $e=\sum_i a_i q^i$ with $a_i\in\{0,\ldots,q-1\}$ in base $q$ expansion. Then the {\bf weight} of $e$ is $w(e)=\sum_i a_i$.
\end{Definition}

\begin{Definition}
Let $f\in R$, $f=\sum_i b_i X^i$. The {\bf weight} of $f$ is $$w(f)=\max\{w(i) : b_i\neq 0\}.$$
\end{Definition}

The next lemma collects some useful facts on the weight and on the degree of the fake Weil descent. The proof is easy and left to the reader.

\begin{Lemma}\label{weights}
Let $e\in\mathbb{Z}_{\geq0}$, $f\in R$. Then:
\begin{enumerate}
\item $w(e)=w(eq)$
\item $w(e)\leq (q-1)\lfloor\log_q(e)+1\rfloor=(q-1)\lceil\log_q(e+1)\rceil$
\item $w(f)\leq (q-1)\lfloor\log_q(\deg(f))+1\rfloor$
\item $\deg(\overline{X^e})=\left\{\begin{array}{ll}     (q-1)n \leq w(e) & \mbox{ if } q^n-1\mid e,\\
                                                                                         w(e \Mod q^n-1)\leq w(e) & \mbox{ otherwise.}
\end{array}\right.$
\end{enumerate}
\end{Lemma}

\begin{remark}\label{degov}
It follows from part 4 of Lemma~\ref{weights} that $\deg(\overline{f})\leq w(f)$. Thus, by part 3, $\deg(\overline{f})\leq (q-1)\lfloor\log_q(\deg(f))+1\rfloor$.
\end{remark}

The first three points of the next lemma are shown in \cite[Lemma 3.2]{HKYY}, 
while the fourth is stated in~\cite[pg. 586]{HKYCrypto}.

\begin{Lemma}\label{addmul}
Let $h_1,h_2,h_3\in R$, let $h\in S$. Then:
\begin{enumerate}
 \item $\overline{h_1+h_2}\equiv_{\max\{\deg(\overline{h_1}),\deg(\overline{h_2})\}}\overline{h_1}+\overline{h_2}$
 \item $\overline{h_1 h_2}\equiv_{\deg(\overline{h_1})+\deg(\overline{h_2})}\overline{h_1}\cdot\overline{h_2}$
 \item there exists $g\in R$ with $\deg(g)<q^n$ such that $h\equiv_{\deg(h)}\overline{g}$
 \item if $\overline{h_1}\equiv_{r}\overline{h_2}$ then $\overline{h_1}\cdot\overline{h_3}\equiv_{\max\{r,\deg(\overline{h_1}\cdot\overline{h_3}),\deg(\overline{h_2}\cdot\overline{h_3})\}}\overline{h_2}\cdot\overline{h_3}$ 
\end{enumerate}
\end{Lemma}

The next lemma is similar \cite[Lemma 4.2]{HKYY}, but we improve the bound by approximately a factor two. 

\begin{Lemma}\label{remainders}
Let $h_1,h_2\in R$ and assume that $\deg(h_2)=d>0$. Let $h_3=h_1 \Mod h_2$. Let $u=(q-1)\lceil\log_q(d)+1\rceil+1$. Assume that $\overline{h_2}\equiv_u 0$. Then $$\overline{h_3}\equiv_{\max\{w(h_1),u\}}\overline{h_1}.$$
\end{Lemma}

\begin{proof} 
Write $h_2=\sum_{i=0}^d b_iX^i$ with $b_d\neq0$. Let $r_e=X^e \Mod h_2$.
Notice that if $h_1=\sum_{i=0}^\delta a_iX^i$, then $h_3=\sum_{i=0}^\delta a_ir_i$. Since $\deg(r_e)<d$, then by Lemma~\ref{weights}.3 and Remark~\ref{degov}, $\deg(\overline{r_e})\leq (q-1)\lfloor\log_q(\deg(r_e))+1\rfloor=(q-1)\lceil\log_q(\deg(r_e)+1)\rceil\leq
(q-1)\lceil\log_q(d)\rceil$. Hence $\overline{h_3}\equiv_{u}\sum_{i=0}^\delta a_i\overline{r_i}$ and $\overline{h_1}\equiv_{w(h_1)}\sum_{i=0}^\delta a_i\overline{X^i}$, by Lemma~\ref{addmul}.1. If $\overline{X^e}\equiv_{\max\{w(e),u\}} \overline{r_e}$, then $\overline{h_3}\equiv_{\max\{w(h_1),u\}}\overline{h_1}$, since by definition $w(h_1)=\max\{w(e) : a_e\neq 0\}$. Thus, we will now show in several steps that
$\overline{X^e}\equiv_{\max\{w(e),u\}} \overline{r_e}$.\\
\underline{Claim 0:} Write $r_{e}=\sum_{i=0}^{d-1}c_i X^i$. 
Then $r_{e+j}=\sum_{i=0}^{d-1} c_i r_{i+j}$ for all $j\geq 0$.\\
\underline{Proof of Claim 0:} By definition $r_{e+j}=X^{e+j}\Mod h_2=X^j r_{e}\Mod h_2$, hence $r_{e+j}=\sum_{i=0}^{d-1} c_iX^{i+j} \Mod h_2=\sum_{i=0}^{d-1} c_i r_{i+j}.$ 
Notice that the last polynomial has degree smaller than $d$, since $\deg(r_i)<d$ for all $i$.\\
\underline{Claim 1:} If $e\in\{0,1,\ldots,qd\}$, then $\overline{X^e}\equiv_u\overline{r_e}$.\\
\underline{Proof of Claim 1:} If $e\leq d-1$, then $r_e=X^e$, and hence $\overline{X^e}\equiv_u\overline{r_e}$. For $e=d$, we have $r_d=\frac{-1}{b_d}\sum_{i=0}^{d-1}b_i X^i$, i.e. $b_d(X^d-r_d)=h_2$ and hence $\overline{X^d}\equiv_u\overline{r_d}$. Now we prove the claim by induction. Assume we have $\overline{X^{e'}}\equiv_u\overline{r_{e'}}$ for all $e'<e$ and $e\leq qd$. Write $r_{e-1}=\sum_{i=0}^{d-1}c_i X^i$. Then $r_e=\sum_{i=0}^{d-1}c_i r_{i+1}$ by Claim 0.
Now $e-1\leq qd-1<q^{\lceil\log_q(d)+1\rceil}$, thus $w(e-1)\leq(q-1)\lceil\log_q(d)+1\rceil$ by Lemma~\ref{weights}.2. Hence, $\deg(\overline{X})+\deg(\overline{X^{e-1}})\leq 1+(q-1)\lceil\log_q(d)+1\rceil=u$, where the inequality follows from Lemma~\ref{addmul}.4. Therefore, 
\begin{align*}
\overline{X^e}\equiv_u&\overline{X}\cdot\overline{X^{e-1}} &\text{ (by Lemma \ref{addmul}.2)}\\
\equiv_u&\overline{X}\cdot\overline{r_{e-1}} &\text{ (by induction and Lemma \ref{addmul}.4, since $w(e)\leq u$)}\\
\equiv_u&\overline{\sum_{i=0}^{d-1}c_i X^{i+1}} &\text{ (by Lemma \ref{addmul}.2)}\\
\equiv_u&\overline{\sum_{i=0}^{d-1}c_i r_{i+1}} &\text{ (by Lemma \ref{addmul}.1 and since $\overline{X^i}\equiv_u\overline{r_i}$ for $i\leq d$)}\\
\equiv_u&\overline{r_{e}}&\text{ (by Claim 0).}\\
\end{align*}
\underline{Claim 2:} If $e$ satisfies $w(e)<u$ and $\overline{X^e}\equiv_{u}\overline{r_e}$, then $\overline{X^{e+1}}\equiv_u\overline{r_{e+1}}$.\\
\underline{Proof of Claim 2:} We have 
\begin{align*}
\overline{X^{e+1}}\equiv_u&\overline{X}\cdot\overline{X^e} &\text{ (by Lemma \ref{addmul}.2)}\\
\equiv_u&\overline{X}\cdot\overline{r_e} &\text{ (by Lemma \ref{addmul}.4)}\\
\equiv_u&\overline{r_{e+1}}.
\end{align*}
\underline{Claim 3:} Assume that $w(e)<u$, $w(e')=1$, $\overline{X^e}\equiv_{u}\overline{r_e}$, and $\overline{X^{e'}}\equiv_{u}\overline{r_{e'}}$ for some $e,e'$. Then $\overline{X^{e+e'}}\equiv_u\overline{r_{e+e'}}$.\\
\underline{Proof of Claim 3:} Write $r_e=\sum_{i=0}^{d-1}c_i X^i$. We have 
\begin{align*}
\overline{X^{e+e'}}\equiv_u&\overline{X^e}\cdot\overline{X^{e'}} &\text{ (by Lemma \ref{addmul}.2)}\\
\equiv_u&\overline{r_e}\cdot\overline{X^{e'}} &\text{ (by Lemma \ref{addmul}.4)}\\
\equiv_u&\overline{\sum_{i=0}^{d-1}c_i X^{i+e'}} &\text{ (by Lemma \ref{addmul}.2)}.
\end{align*}
But $\overline{X^{e'}}\equiv_{u}\overline{r_{e'}}$. If $d=1$, we are done, so assume $d>1$. By Claim 2, $\overline{X^{e'+1}}\equiv_u\overline{r_{e'+1}}$. Since $w(e'+1)<u$, we can apply Claim 2 again. By repeated application of Claim 2, we get $\overline{X^{e'+i}}\equiv_u\overline{r_{e'+i}}$ for $i\leq d-1$ since $w(i+e')\leq w(i)+w(e')\leq (q-1)\lceil\log_q(d)\rceil+1\leq u$. Thus $$\overline{\sum_{i=0}^{d-1}c_i X^{i+e'}}\equiv_u\overline{\sum_{i=0}^{d-1}c_i r_{i+e'}}\equiv_u\overline{r_{e+e'}}$$ by Claim 0.\\
\underline{Claim 4:} If $e=mq^k$ for some $k\geq0$, $1\leq m<q$, then $\overline{X^e}\equiv_u\overline{r_e}$.\\
\underline{Proof of Claim 4:} We will prove the claim by induction on $k$ and $m$. If $k=0$ then the statement is true for all $m$ by Claim 1. We now assume that the statement holds for $e=mq^{k-1}$ for all $m$ and we show it for $e=q^k$. Letting $e=(q-1)q^{k-1}$ and $e'=q^{k-1}$ in Claim 3, we get $\overline{X^{q^k}}\equiv_u\overline{r_{q^k}}$. To complete the proof, assume that the statement holds for $e=\ell q^k$ for $1\leq \ell\leq m-1$ and show it for $e=mq^k$.
By Claim 3 with $e=(m-1)q^k$ and $e'=q^k$, we get $\overline{X^{mq^k}}\equiv_u\overline{r_{mq^k}}$.\\
\underline{Claim 5:} If $e$ satisfies $w(e)\leq u$, then $\overline{X^e}\equiv_u\overline{r_e}$.\\
\underline{Proof of Claim 5:} We will prove the claim by induction on $w(e)$. Let $w(e)=1$. Then $e=q^k$ for some $k\geq0$ and by Claim 4, $\overline{X^e}\equiv_u\overline{r_e}$. So assume $\overline{X^{e'}}\equiv_u\overline{r_{e'}}$ for $w(e')<w(e)$. We can write $e=e_1+e_2$ such that $w(e)=w(e_1)+w(e_2)$ and $w(e_2)=1$ (e.g. let $e_2=q^{\lfloor\log_q(e)\rfloor}$). Then by Claims 3 and 4 and by induction, we have $\overline{X^e}\equiv_u\overline{r_e}$.\\
\underline{Claim 6:} $\overline{X^e}\equiv_{\max\{w(e),u\}}\overline{r_e}$.\\
\underline{Proof of Claim 6:} As before, write $e=e_1+e_2$ such that $w(e)=w(e_1)+w(e_2)$ and $w(e_2)=1$. If $w(e_1)<u$, the thesis follows by Claim 3 and Claim 5. If $w(e_1)=u$, then by Claim 5, $\overline{X^{e_1}}\equiv_u\overline{r_{e_1}}$. Thus we have 
\begin{align*}
\overline{X^e}=\overline{X^{e_1+e_2}}&\equiv_{\max\{w(e),u\}}\overline{X^{e_1}}\cdot\overline{X^{e_2}} &\text{ (by Lemma \ref{addmul}.2)}\\
&\equiv_{\max\{w(e),u\}}\overline{r_{e_1}}\cdot\overline{X^{e_2}} &\text{ (by Lemma \ref{addmul}.4)}\\
&\equiv_{\max\{w(e),u\}}\overline{r_{e_1+e_2}}=\overline{r_e}.
\end{align*}
Here the last equality can be proved using the same argument as in Claim 3, noticing that by Claim 4, $\overline{X^{e_2}}\equiv_u\overline{r_{e_2}}$. This proves Claim 6 if $w(e)\leq u+1$. Proceed by induction on $w(e)$.
Letting $e=e_1+e_2$ with $w(e)=w(e_1)+w(e_2)$ and $w(e_2)=1$ and assuming by induction that $\overline{X^{e_1}}\equiv_{\max\{w(e),u\}}\overline{r_{e_1}}$, the same argument as above shows that $\overline{X^e}\equiv_{\max\{w(e),u\}}\overline{r_{e}}$.
\end{proof}

\begin{remark}
The factor 2 improvement in the previous lemma is achieved mainly in Claims 2 and 3. It follows from the idea that instead of multiplying a polynomial by $X^{e'}$ and then reducing $\Mod{h_2}$, we can repeatedly ($e'$ times) multiply by $X$ and reduce $\Mod{h_2}$ at every step, and thereby the intermediate polynomials have lower degrees.
\end{remark}

The next example shows that the bound of Lemma~\ref{remainders} is sharp.

\begin{example}
Let $k=\mathbb{F}_{2^2}=\mathbb{F}_{2}[t]/(t^2+t+1)$. Let $h_1:=X^3+tX^2+X+t^2\in R=k[X]$, and let $h_2:=X+1\in R$. Then $h_3=1$ and $\overline{h_1}-\overline{h_3}=X_0X_1+tX_1+X_0+t\in S=k[X_0,X_1]$ and $u=2$.
\end{example}

The following proposition is similar to \cite[Proposition 4.3]{HKYY}, but yields a tighter bound, due to the improvement in Lemma~\ref{remainders}.

\begin{Proposition}\label{GCDProp}
Let $\mathcal{F}=\{f\}$ with $\deg(f)>0$. Let $u=(q-1)\lceil\log_q(\deg(f))+1\rceil+1$ and let 
$g=\gcd(f,X^{q^n}-X)$. Then $\overline{g}\in V_{\overline{\mathcal{F}}_f,u}$.
\end{Proposition}
\begin{proof}
Write $V_u$ for $V_{\overline{\mathcal{F}}_f,u}$.
We use the Euclidean algorithm to compute the GCD of $f$ and $X^{q^n}-X$. It works as follows: At every step $k$, the Euclidean algorithm computes the remainder $g_k$ as $g_k:=g_{k-2} \Mod g_{k-1}$, where $g_0=f$ and $g_{-1}=X^{q^n}-X$. The algorithm terminates when $g_k=0$ for some $k$. Then $g_{k-1}=g=\gcd(f,X^{q^n}-X)$.

We claim that for every polynomial $g_j$ with $j\geq 1$, one has $\overline{g_j}\in V_u$, that is $\overline{g_j}\equiv_u 0$. We proceed by induction on $j\geq 1$. For $j=1$, the algorithm computes $g_1:=X^{q^n}-X \Mod f$. By Lemma \ref{remainders}, letting $h_1=X^{q^n}-X$, $h_2=f$, and $h_3=g_1$, we obtain $\overline{g_1}\equiv_u\overline{X^{q^n}-X}=0$, since $w(h_1)=1\leq u$.

Assume now that $\overline{g_i}\equiv_u 0$ for $1\leq i\leq j-1$. By Lemma~\ref{remainders}, letting $h_1=g_{j-2}$, $h_2=g_{j-1}$, and $h_3=g_j$, we get $\overline{g_j}\equiv_u\overline{g_{j-2}}\equiv_u 0$, since $w(g_{j-2})\leq u$. 
Notice that $\deg(g_{j-1})>0$, except possibly for $j=k$. If $\deg(g_{k-1})=0$, then $\overline{g_{k-1}}\in V_u$ is invertible, hence $\overline{g_k}\equiv_u 0$.
\end{proof}

The following theorem corresponds to~\cite[Theorem 4.5]{HKYY} and~\cite[Theorem 1]{HKYCrypto}. The proof is very similar, but we use our improved bound. We will also use the following lemma.

\begin{Lemma}[\cite{HKYY}, Lemma 3.3]\label{ideals}
Let $h\in R$. Then $h\in I$ if and only if $\overline{h}\in\overline{I}$, where $I\subseteq R$ is the ideal generated by $\mathcal{F}_f$ and $\overline{I}\subseteq S$ is the ideal generated by $\overline{\mathcal{F}}_f$.
\end{Lemma}

In the next theorem we assume that $\mathcal{F}$ does not contain any constants. In fact, if $\mathcal{F}$ contains the zero polynomial, then it can be removed without affecting the last fall degree of $\mathcal{F}$. If $\mathcal{F}$ contains a nonzero constant, then its last fall degree is zero.

\begin{Theorem}\label{main}
Let $k$ be a finite field of cardinality $q^n$, and let $\mathcal{F}\subset R\setminus k$ be finite. Let $d\in\mathbb{Z}_{>0}$ be the smallest integer such that $\deg(f)\leq d$ for some $f\in\mathcal{F}$ and such that for all $f\in\mathcal{F}$, we have $\deg(\overline{f})\leq (q-1)\lceil\log_q(d)+1\rceil+1$. Then $$d_{\mathcal{F}'_f}\leq (q-1)\lceil\log_q(d)+1\rceil+1.$$
\end{Theorem}

\begin{proof}
By Remark \ref{RelateWeil}, $d_{\mathcal{F}'_f}\leq\max\{d_{\overline{\mathcal{F}}_f},q,\deg(\mathcal{F}')\}$, so we will study $d_{\overline{\mathcal{F}}_f}$, and we define $\equiv_i$ with respect to $\overline{\mathcal{F}}_f$.
Let $u=(q-1)\lceil\log_q(d)+1\rceil+1$. Let $g=\gcd(\mathcal{F}\cup\{X^{q^n}-X\})$ and let $f\in\mathcal{F}$ with $0<\deg(f)\leq d$. Then $\overline{g}\in V_u$ by Proposition~\ref{GCDProp}, since $V_{\overline{\mathcal{F}}_f,i}\supseteq V_{{\overline{\{f\}}}_f,i}$ for all $i$. 

Let $h\in\overline{I}=(\overline{\mathcal{F}}_f)$. By Lemma \ref{addmul}.3, there exists $h_1\in k[X]$ with $\deg(h_1)<q^n$ such that $\overline{h_1}\equiv_{\deg(h)}h$. Thus $\overline{h_1}\in\overline{I}$ and by Lemma \ref{ideals}, $h_1\in I$. Thus $h_1=0 \Mod g$, and by Lemma \ref{remainders} $\overline{h_1}\equiv_{\max\{w(h_1),u\}}0$. Hence $h\in V_{\max\{\deg(h),u\}}$, or equivalently $h\equiv_{\max\{\deg(h),u\}}0$, since $w(h_1)=\deg(\overline{h_1})\leq\deg(h)$. Thus by Definition \ref{LastFallDegree}, $d_{\overline{\mathcal{F}}_f}\leq u$. Now by Remark~\ref{degrees}, $\deg(\mathcal{F}')=\deg(\overline{\mathcal{F}})\leq u$, and therefore, $d_{\mathcal{F}'_f}\leq\max\{u,q\}$.
\end{proof}

\begin{example}\label{boundachieved}
Let $k=\mathbb{F}_{2^5}=\mathbb{F}_{2}[t]/(t^5+t^2+1)$. Let $\mathcal{F}=\{f_1,f_2\}$ where $f_1=t^{16}X^{11}+1$ and $f_2=tX^{31}+1$. Then $w(f_1)=3$ and $w(f_2)=5$, so $d=11$ and $(q-1)\lceil\log_q(d)+1\rceil+1=6$. Theorem~\ref{main} tells us that $d_{\mathcal{F}'_f}\leq6$. Performing a \groebner basis algorithm on $\mathcal{F}'_f$ (in degree reverse lexicographic order) in Magma~\cite{Magma} in fact gives us a solving degree of 6.
\end{example}

The theorem allows us in particular to give an upper bound on the last fall degree of HFE. In the next result, we refer to the version of HFE from~\cite{P}.

\begin{corollary}
Let $k$ be a finite field of cardinality $q^n$ and let $\mathcal{F}=\{f\}$ where 
$$f=\sum_{i,j}\beta_{i,j}X^{q^{\theta_{ij}}+q^{\varphi_{ij}}}+\sum_\ell \alpha_\ell X^{q^{\zeta_\ell}}+\mu\in k[X]$$
and $\deg(f)\leq q^t$ with $\theta_{ij},\varphi_{ij},\zeta_\ell\in\mathbb{Z}$.
Then $$d_{\mathcal{F}'_f}\leq (q-1)(t+1)+1.$$
\end{corollary}

Example~\ref{boundachieved} shows that in general, the bound of Theorem~\ref{main} can be reached. In the case of HFE polynomials however, our bound is still larger (by approximately a factor 2) than the (heuristic) bound of \cite{DingHodges} and the experimental results of \cite{FJ} and further work needs to be done to close the gap between experiments and rigorous bounds.

\bibliographystyle{alpha}

\begin{thebibliography}{HKYY18}

\bibitem[BCP97]{Magma}
Wieb Bosma, John Cannon, and Catherine Playoust.
\newblock The {M}agma algebra system. {I}. {T}he user language.
\newblock {\em Journal of Symbolic Computation}, 24(3-4):235--265, 1997.
\newblock Computational algebra and number theory (London, 1993).

\bibitem[BS87]{BS}
David Bayer and Michael Stillman.
\newblock A criterion for detecting m-regularity.
\newblock {\em Inventiones Mathematicae}, 87:1--12, 1987.

\bibitem[CG17]{CG}
Alessio Caminata and Elisa Gorla.
\newblock Solving multivariate polynomial systems and an invariant from
  commutative algebra.
\newblock {P}reprint, 2017.

\bibitem[CKPS00]{CKPS00}
Nicolas Courtois, Alexander Klimov, Jacques Patarin, and Adi Shamir.
\newblock Efficient algorithms for solving overdefined systems of multivariate
  polynomial equations.
\newblock In {\em Advances in Cryptology - {EUROCRYPT} 2000, International
  Conference on the Theory and Application of Cryptographic Techniques, Bruges,
  Belgium, May 14-18, 2000, Proceeding}, pages 392--407, 2000.

\bibitem[DH11]{DingHodges}
Jintai Ding and Timothy~J. Hodges.
\newblock Inverting {HFE} systems is quasi-polynomial for all fields.
\newblock In Phillip Rogaway, editor, {\em Advances in Cryptology -- CRYPTO
  2011}, pages 724--742, Berlin, Heidelberg, 2011. Springer Berlin Heidelberg.

\bibitem[Fau99]{F4}
Jean~Charles Faug\`{e}re.
\newblock A new efficient algorithm for computing {Gr\"{o}bner} bases ({F4}).
\newblock {\em Journal of Pure and Applied Algebra}, 139(1):61 -- 88, 1999.

\bibitem[Fau02]{F5}
Jean~Charles Faug\`{e}re.
\newblock A new efficient algorithm for computing {Gr\"{o}bner} bases without
  reduction to zero ({F5}).
\newblock In {\em Proceedings of the 2002 International Symposium on Symbolic
  and Algebraic Computation}, ISSAC '02, pages 75--83, New York, NY, USA, 2002.
  ACM.

\bibitem[FJ03]{FJ}
Jean-Charles Faug\`ere and Antoine Joux.
\newblock {Algebraic Cryptanalysis of Hidden Field Equation (HFE) Cryptosystems
  Using Gr\"obner Bases}.
\newblock In D.~Boneh, editor, {\em Advances in Cryptology - CRYPTO 2003},
  number 2729 in Lecture Notes in Computer Science, pages 44--60. Springer,
  Berlin, Heidelberg, 2003.

\bibitem[HKY15]{HKYCrypto}
Ming-Deh~A. Huang, Michiel Kosters, and Sze~Ling Yeo.
\newblock Last fall degree, {HFE}, and {Weil} descent attacks on {ECDLP}.
\newblock In Rosario Gennaro and Matthew Robshaw, editors, {\em Advances in
  Cryptology -- CRYPTO 2015}, pages 581--600, Berlin, Heidelberg, 2015.
  Springer Berlin Heidelberg.

\bibitem[HKYY18]{HKYY}
Ming-Deh~A. Huang, Michiel Kosters, Yun Yang, and Sze~Ling Yeo.
\newblock On the last fall degree of zero-dimensional {Weil} descent systems.
\newblock {\em Journal of Symbolic Computation}, 87:207 -- 226, 2018.

\bibitem[Laz83]{Lazard}
Daniel Lazard.
\newblock Gr{\"o}bner bases, gaussian elimination and resolution of systems of
  algebraic equations.
\newblock In J.~A. van Hulzen, editor, {\em Computer Algebra}, pages 146--156,
  Berlin, Heidelberg, 1983. Springer Berlin Heidelberg.

\bibitem[Pat96]{P}
Jacques Patarin.
\newblock {Hidden Fields Equations (HFE) and Isomorphisms of Polynomials (IP):
  Two New Families of Asymmetric Algorithms}.
\newblock In U.~Maurer, editor, {\em Advances in Cryptology -- EUROCRYPT 1996},
  volume 1070 of {\em Lecture Notes in Computer Science}, pages 33--48.
  Springer, Berlin, Heidelberg, 1996.

\end{thebibliography}

\end{document}